\documentclass[letterpaper, 10 pt, conference]{ieeeconf}
\IEEEoverridecommandlockouts  
\usepackage{graphics,epstopdf,wrapfig}
\usepackage{epsfig}
\usepackage{extarrows}
\usepackage{algorithm}
\usepackage{algorithmic}
\usepackage{multirow}
\usepackage{subfigure}
\usepackage{amsmath}
\usepackage{tabularx,booktabs}
\usepackage{mathrsfs}
\usepackage{epsfig}
\usepackage{amssymb,latexsym,amsfonts,amsmath}
\usepackage{graphicx}
\usepackage{subfigure}
\usepackage{etex}
\usepackage{color}
\usepackage{cite}
\usepackage{tikz,verbatim}
\usetikzlibrary{automata,positioning,shapes,arrows}
\usetikzlibrary{shapes,snakes}
\tikzset{>=latex}
\tikzstyle{block} = [draw, rectangle, minimum size=3em]
\usepackage{epstopdf}
\usepackage{xcolor}
\usepackage[curve]{xypic}

\usepackage[justification=centering]{caption}
\usepackage{float}
\usepackage{subfigure}
\newtheorem{theorem}{Theorem}

\newtheorem{proposition}{Proposition}
\newtheorem{definition}{Definition}

\newtheorem{remark}{Remark}
\newtheorem{example}{Example}

\allowdisplaybreaks

\title{\bf Synthesis of Insertion Functions to Enforce Decentralized and Joint Opacity Properties of Discrete-event Systems}

\author{Bo~Wu, Jin~Dai and Hai~Lin
	\thanks{The partial support of the National Science Foundation (Grant No. CNS-1446288, ECCS-1253488, IIS-1724070) and of the Army Research Laboratory (Grant No. W911NF- 17-1-0072) is gratefully acknowledged.}
	\thanks{The authors are with the Department of Electrical Engineering, University of Notre Dame, Notre Dame,
		IN, 46556 USA. {\tt\small bwu3@nd.edu, jdai1@nd.edu, hlin1@nd.edu}}}

\begin{document}
\maketitle
\thispagestyle{empty}
\pagestyle{empty}

\begin{abstract}
Opacity is a confidentiality property that characterizes the non-disclosure of specified secret information of a system to an outside observer. In this paper, we consider the enforcement of opacity within the discrete-event system formalism in the presence of multiple intruders. We study two cases, one without coordination among the intruders and the other with coordination. We propose appropriate notions of opacity corresponding to the two cases, respectively, and propose enforcement mechanisms for these opacity properties based on the implementation of insertion functions, which manipulates the output of the system by inserting fictitious observable events whenever necessary. The insertion mechanism is adapted to the decentralized framework to enforce opacity when no coordination exists. Furthermore, we present a coordination and refinement procedure to synthesize appropriate insertion functions to enforce opacity when intruders may coordinate with each other by following an intersection-based coordination protocol. The effectiveness of the proposed opacity-enforcement approaches is validated through illustrative examples.
\end{abstract}


\section{Introduction}\label{sect:introduction}
Security and privacy have become important issues in the design of cyber and cyber-physical systems \cite{bishop2003computer}. In this paper, we focus our study on {\it opacity} \cite{mazare2004using}, which is a confidentiality property  that justifies whether a given system's confidential information (denoted as {\it ``secret"}) is kept uncertain from an external observer (termed as an {\it intruder}). Since many security and privacy properties, such as anonymity \cite{kumari2014more}, trace-based non-interference \cite{hadj2005verification} and secrecy \cite{alur2006preserving,rabbachin2015wireless}, can be expressed in terms of opacity \cite{lin2011opacity}, it has emerged as an active research topic in the computer science and control literature, see, e.g., \cite{jacob2016overview} and the references therein.

Motivated by the fact that many engineering systems are inherently event-driven, we consider opacity issues in the framework of discrete-event systems (DES) \cite{cassandras2008introduction}. An opacity problem is generally formulated as follows in the context of DES: (i) the system is modeled as a Petri net \cite{bryans2005modelling,tong2017verification} or a finite automaton \cite{saboori2007notions}; (ii) the system possesses a secret that is expected to be hidden from an intruder; (iii) the intruder is an observer with full knowledge of the system's structure but can only observe part of the system's behavior. The system is said to be opaque with respect to the given secret if the intruder can never determine unambiguously that the secret has occurred based on its observation of the system's behaviors. More specifically, if opacity of the system holds, then for any behavior that may reveal the secret (termed {\it secret behavior}), there exists at least one behavior that does not reveal the secret (termed {\it non-secret behavior}) which shares the same observation of the secret behavior to the intruder; thus, the intruder can never be sure if the secret or the non-secret has occurred. Depending on how the secret is represented, various notions of opacity have been introduced in the literature, and considerable amount of research efforts has been devoted to the formal verification of language-based opacity \cite{lin2011opacity}, current-state opacity \cite{saboori2014current}, initial-state opacity \cite{saboori2007notions}, $K$-step opacity \cite{saboori2011verification} and infinite-step opacity \cite{saboori2012verification}. 

In case the system fails to be opaque, formal methods have also been proposed to enforce opacity. Design of opacity-enforcing supervisory controllers for restricting the system's behavior to ensure opacity by disabling any behavior that will reveal the secret has been studied extensively in literature \cite{saboori2012opacity,dubreil2010supervisory,yin2016uniform,ben2016opaque}. Nevertheless, the supervisory control approach is not suitable for situations where the system must execute its full behavior. A runtime mechanism was developed in \cite{falcone2015enforcement} to enforce $K$-step opacity based on delaying the output; however, this method only ensured opacity of secrets whose time duration was of concern. Rather than supervisory control approach, we consider enforcement strategies that do not alter the behavior of the system and instead ensure opacity by appropriately manipulating the system's output information whenever necessary. One of the enforcement techniques was implemented by a dynamic observer in \cite{cassez2012synthesis}. However, the intermittent loss of observability of certain events may render the observation of the intruder inconsistent with its knowledge of the system and its original observation capabilities, which may remind the intruder of the existence of the opacity-enforcement mechanism. Wu and Lafortune \cite{wu2014synthesis} proposed an enforcement mechanism based on insertion of fictitious observable events at the system's output; the inserted events were observationally equivalent to the system's genuine observable events from the intruder's perspective, therefore making the intruder confused.


Recent advances in communication and network technologies have made large-scale systems with spatially-decentralized and/or distributed architectures more widely used in the application; therefore, opacity problems for DES with decentralized structure are of both academic and practical importance. For instance, for a cryptosystem that can be observed by users of multiple security levels, opacity should be guaranteed in such a way that: (i) users with lower security level can never infer any information which can only be accessed by users of high security level \cite{hadj2005verification}; (ii) even if a user has a high security level, it is still not able to infer any private information that is possessed by a user with low security level \cite{badouel2007concurrent}. Compared to the fruitful contributions that have been made to opacity problems in the presence of a single intruder, limited studies have been made to the cases where the system can be observed by multiple intruders. Badouel et al. \cite{badouel2007concurrent} considered multiple intruders, each of them having its own observation mapping and the secret of interest. The system therein was said to be concurrently opaque if all secrets can be kept safe. A different notion termed as ``joint opacity" was proposed in \cite{wu2013comparative}, in which a team of intruders collaborated through a coordinator to infer the secret of common interest. Paoli and Lin \cite{paoli2012decentralized} studied decentralized opacity issues with and without coordination among the intruders. Nevertheless, to the best of the authors' knowledge, most of the existing results are established on opacity verification problems while no prior work has been proposed to investigate opacity-enforcement problems in the presence of multiple intruders.

We are therefore motivated to study opacity-enforcement problems of DES that can be observed by multiple intruders. By modeling the system as a finite automaton, we assume that each intruder has full prior knowledge of the system model but can only partially observe the behavior of the system. We investigate opacity problems in two cases, one assuming no coordination among the intruders and the other assuming that the intruders may coordinate with each other. We adopt the enforcement mechanism based on insertion functions to assure decentralized opacity when no coordination exists among the intruders. Furthermore, we study the enforcement of joint opacity when the intruders may coordinate via an intersection-based protocol. Facing the coordinated intruders, we propose a centralized coordination and refinement procedure to construct local insertion functions associated with each intruder's observation capabilities such that joint opacity can be guaranteed.

The remainder of this paper is organized as follows. We present the system model and relevant concepts of opacity problems in DES and the insertion-based opacity-enforcement mechanism of DES in Section II. We study the opacity-enforcement problem of DES in the presence of multiple non-coordinating intruders and compute appropriate insertion functions for each intruder in Section III. Under the assumption that the intruders may coordinate via an intersection-based protocol, we introduce the notion of joint opacity in Section IV and develop enforcement schemes for joint opacity by incorporating the synthesis of local insertion functions with centralized coordination. Finally, we end this paper with concluding remarks and discussion of future research directions in Section V.

\section{Opacity of Discrete-event Systems}\label{sect:Preliminaries}

\subsection{Preliminaries of Discrete-event Systems}\label{subsect:Automata models}
The following notation and concepts are standard in the DES literature \cite{cassandras2008introduction}. For a finite alphabet $E$ of event symbols, $|E|$ and $2^E$ denote the cardinality and power set of $E$, respectively. $E^*$ stands for the set of all finite strings over $E$ plus the empty string $\epsilon$. A subset of $L\subseteq E^*$ is called a language over $E$. The prefix closure of $L$ is defined by $\overline{L}=\{s\in E^*|(\exists t\in E^*)[st\in L]\}$. $L$ is said to be {\it prefix-closed} if $L=\overline{L}$.

We consider the DES modeled as a {\it non-deterministic finite automaton} (NFA) $G=(X,E,f,X_0)$, where $X$ is the finite set of states, $E$ is the finite set of events, $f: X\times E\to 2^X$ is the (partial) transition function, $X_0\subseteq X$ is the set of initial states. The transition function $f$ can be extended to $X\times E^*$ in the natural way \cite{cassandras2008introduction}.Given a set $X'\subseteq X$ of states, the language generated by $G$ from $X'$ is defined by $L(G, X')=\{s\in E^*\vert (\exists x'\in X')[f(x',s)!]\}$, where $f(x',s)!$ means that the transition $f(x',s)$ is defined. The generated behavior of $G$ is then given by $L(G,X_0)$. We write $L(G)$ for simplicity if $X_0$ is clear from the context.

In general, the system $G$ can only be partially observed. Towards this end, $E$ is partitioned into two disjoint subsets, i.e., $E=E_o\dot\cup E_{uo}$, where $E_o$ is the set of observable events and $E_{uo}$ is the set of unobservable events. The presence of partial observation is captured by the natural projection $P:E^* \to E_o^*$, which is defined as:
\begin{equation}
P(\epsilon)=\epsilon, \mbox{ and }
P(se)=
\begin{cases}
P(s)e,  & \mbox{if } e\in E_o \\
P(s), &\mbox{if } e\in E_{uo}
\end{cases}    
\end{equation}
for all $s\in E^*$ and $e\in E$. The inverse projection of $P$ is defined as $P^{-1}(t)=\left\{s\in E^*\vert P(s)=P(t)\right\}$ for $t\in E_o^*$.

\subsection{Current-state Opacity of Discrete-event Systems}\label{sect:pro-definition}
The ingredients of an opacity-enforcement problem in DES include: (i) $G$ has a secret; (ii) the intruder is an observer with full knowledge of the structure of $G$; (iii) the intruder can only observe the behavior of $G$ partially due to its limited observation capabilities $E_o$. With the prior knowledge of $G$, the intruder can infer the system's evolution by constructing estimates on the basis of online observations. Depending on how the secret is defined, various notions of opacity have been extensively studied in the literature. In this paper, we define the secret to be a set of states of $G$ and consider the notion of {\it current-state opacity}. Intuitively, the system $G$ is current-state opaque if for any secret behavior that visits a secret state, there always exists a non-secret behavior of $G$ that visits a non-secret state while the intruder cannot distinguish
between these two behaviors. Formally, current-state opacity is defined as follows.

\begin{definition}[Current-state Opacity (CSO)]
Given the set of observable events $E_o\subseteq E$, the set of {\it secret states} $X_S\subseteq X$ and the set of {\it non-secret states $X_{NS}\subseteq X$}, the system $G=(X,E,f,X_0)$ is said to be current-state opaque  with respect to $E_o$, $X_S$ and $X_{NS}$ if 
\begin{equation}
\begin{split}
    & (\forall x_0\in X_0)(\forall t\in L(G,x_0): f(x_0,t)\in X_S) \Rightarrow (\exists x_0'\in X_0)\\
    & (\exists t'\in L(G,x_0'))[(f(x_0',t')\in X_{NS})\land(P(t')=P(t))].
\end{split}
\end{equation}
\end{definition}

\begin{remark}
We assume without loss of generality in the rest of this paper that the set of non-secret states is the complement of the secret state set, i.e., $X_{NS}=X\setminus X_S$.
\end{remark}

\begin{remark}
According to \cite{wu2013comparative}, other notions of opacity, specifically language-based opacity, initial-state opacity and initial-and-final-state opacity, can all be transformed to CSO in polynomial time. Thus, our proposed enforcement approach for CSO of DES applies to the enforcement of other opacity notions as well.
\end{remark}

\subsection{Event Insertion Mechanism}\label{sect:problem formulation}

In \cite{wu2014synthesis}, the authors proposed an opacity-enforcement mechanism based on the implementation of insertion functions when the system $G$ fails to be CSO.  As shown in Fig. \ref{fig:fig9}, an insertion function serves as a special monitoring interface between the system and the intruder. The insertion function receives an output behavior in $s\in P[L(G)]$ and inserts fictitious observable events before $s$ is observed whenever the intruder may infer the occurrence of the secret from $s$. It is worth pointing out that the intruder cannot distinguish inserted observable events from the system's genuine observable events.

\begin{figure}[H]
\centering
\includegraphics[scale=0.4]{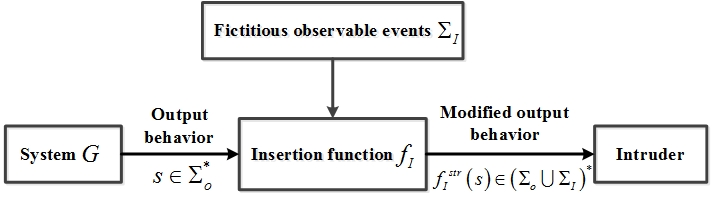}\\
\caption{Opacity-enforcement based on event insertion.}\label{fig:fig9}
\vspace{-2.5mm}
\end{figure}

For the purpose of clear presentation, we associate each inserted event with an ``insertion label" $I$, and the set of inserted events is denoted by ${E_I} = \{ {e_I}:e \in {E_o}\} $.
Formally, the basic structure of an insertion function  is defined as a (possibly partial) mapping ${f_I}$: $E_o^ *\times E_o \to E_I^ * {E_o}$ that outputs a string with necessarily inserted events based on the system's historical and current output behavior. Given a string $se_o\in P[L(G)]$ that has been observed by the insertion function, the output behavior of the insertion function before the occurrence of $e_o$
is defined as ${f_I}(s,{e_o}) = {s_I}{e_o}$ where ${s_I} \in E_I^* $ is the inserted string. In the sequel, we assume additionally that length of $s_I$ is bounded from above. To determine the complete modified output from the insertion function, we define recursively an induced insertion function $f_I^{str}$ from $f_I$: $f_I^{str}(\epsilon ) = \epsilon $ and $f_I^{str}(s_n) = {f_I}(\epsilon ,e_1){f_I}(e_1,e_2) \cdots {f_I}(e_1e_2e_{n - 1},e_n)$ where $s_n = e_1e_2 \cdots e_n \in E_o^ * $. 

The modified output $L_{out}$ of the system $G$ under the impact of the insertion function $f_I$ is then given by 
\begin{equation}
\begin{split}
L_{out}:&=f_I^{str}(P[L(G)]) \\
&=\{ \tilde s \in {(E_i^ * {E_o})^ * }\vert\exists s \in P[L(G)]: \tilde s = f_I^{str}(s)\}.
\end{split}
\end{equation}
To pursue succinct notations, we use $f_I$ and $f_I^{str}$ interchangeably in the sequel. Specifically, in this paper we are looking for the insertion functions that satisfy the \empty{private enforceability} \cite{wu2016enhancing}.
\begin{definition}[Private Enforceability]\label{def:private enforceability}
    Given a DES $G$ and the observation mask $P$, an insertion function $f_I$ is privately enforcing if (i) admissibility: $\forall se_o\in P[L(G)], s\in E_o^*, e_o\in E_o$, $\exists s_I\in E_I^*$ such that $f_I(s,e_o)=s_Ie_o$; (ii) private safety: $L_{out}\subseteq L_{safe}=P[L(G)]\backslash (P[L(G)]\backslash P(L_{NS}))E^*_o$, where $L_{NS}=\{t\in L(G,X_0)|\exists x_i\in X_0, f(x_i,t)\cap X_{NS}\neq\emptyset\}$.
\end{definition}

Intuitively, the admissibility requires that the insertion function $f_I$ is well defined on all the strings from $P[L(G)]$. The private safety requirement restricts the modified output $L_{out}$ from $f_I$ to the non-secret behavior of the system $L_{safe}$, which is the set of projected strings that never reveal the secret. Therefore, by the definition, a private enforcing insertion function guarantees CSO.

\section{Enforcement of Decentralized Opacity via Insertion Functions}

In this section, we focus our study on opacity problems of DES with a decentralized architecture. Specifically, we extend the investigation of opacity-enforcement strategies to the case in which the system $G$ can be observed by multiple intruders as shown in Fig. \ref{fig:decentralized}.

\begin{figure}[H]
\centering
\includegraphics[scale=0.4]{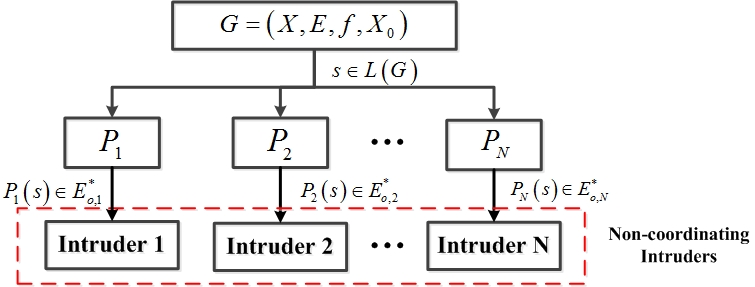}\\
\caption{The DES $G$ observed by non-coordinating intruders.}\label{fig:decentralized}
\vspace{-3.5mm}
\end{figure}

We first consider the case where no coordination exists among intruders. Let $I_i$, $i\in\mathcal{N}=\{1,2,\ldots,N\}$ denote a team of $N$ intruders. Similar to the centralized scenario, each intruder has a complete prior knowledge of the system $G$. Intruder $I_i$ is associated with the locally observable events $E_{o,i}\subseteq E$, $i\in\mathcal{N}$. The partial observation for $I_i$ is characterized by the projection $P_i: E^*\to E_{o,i}^*$ when no insertion function exists. The property of {\it decentralized current-state opacity} is formally defined as follows.

\begin{definition}[Decentralized CSO (D-CSO)]\label{def:DCSO}
Given the set of observable events $E_{o,i}\subseteq E$ for intruder $I_i$, $i\in \mathcal{N}$, the secret state set $X_S$, and the non-secret state set $X_{NS}$, the system $G=(X,E,f,X_0)$ is said to be decentralized current-state opaque with respect to $E_{o,i}$ $i\in \mathcal{N}$, $X_S$ and $X_{NS}$ if 
\begin{equation}
\begin{split}
    & (\forall i\in\mathcal{N})(\forall x_{0,i}\in X_0)(\forall t_i\in L(G,x_{0,i}): f(x_{0,i},t_i)\in X_S)  \\
    &\Rightarrow (\exists x_{0,i}'\in X_0)(\exists t_i'\in L(G,x_{0,i}'))[(f(x_{0,i}',t_i')\in X_{NS})\land \\
    & (P_i(t_i')=P_i(t_i))].
\end{split}
\end{equation}
\end{definition}

The D-CSO of $G$ suggests that any secret state in $X_S$ be not inferred by any one of the intruders. It follows from Definition \ref{def:DCSO} that D-CSO can be viewed as a decentralized counterpart of CSO, which implies that enforcing D-CSO for $G$ with multiple intruders is equivalent to enforcing CSO with respect to each individual intruder $I_i$, $i\in\mathcal{N}$. Motivated by this fact, we can synthesize local opacity-enforcing insertion function $f^i$ for $I_i$, $i\in\mathcal{N}$ independently.

We illustrate the idea of synthesizing appropriate insertion functions for each intruder by the following example.

	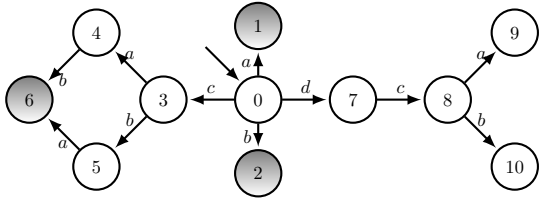
\begin{figure}[ht]
		\centering	
		\begin{tikzpicture}[shorten >=1pt,node distance=1.8 cm,on grid,auto, bend angle=20, thick,scale=.7, every node/.style={transform shape}]
		\node[state] (s_0)   {$0$};
		\node[state,shade] (s_1) [above= 1.4 cm of s_0] {$1$};
		\node[state,shade] (s_2) [below= 1.4 cm of s_0] {$2$};
		\node[state] (s_3) [left=of s_0] {$3$};
		\node[state] (s_4) [above left=of s_3] {$4$};
		\node[state] (s_5) [below left=of s_3] {$5$};
		\node[state,shade] (s_6) [below left=of s_4] {$6$};
		\node[state] (s_7) [right=of s_0] {$7$};
		\node[state] (s_8) [right=of s_7] {$8$};
		\node[state] (s_9) [above right=of s_8] {$9$};
		\node[state] (s_10) [below right=of s_8] {$10$};
		\draw [->] (-1,1) -- (s_0);
		\path[->]
		(s_0) edge node [pos=0.5, left]{$a$} (s_1)
		(s_0) edge node [pos=0.5,  left]{$b$} (s_2)
		(s_0) edge node [pos=0.5,  above]{$c$} (s_3)
		(s_0) edge node [pos=0.5, above]{$d$} (s_7)
		(s_3) edge node [pos=0.5,  above]{$a$} (s_4)
		(s_3) edge node [pos=0.5,  above]{$b$} (s_5)
		(s_4) edge node [pos=0.5,  below]{$b$} (s_6)
		(s_5) edge node [pos=0.5,  below]{$a$} (s_6)
		(s_7) edge node [pos=0.5, above]{$c$} (s_8)
		(s_8) edge node [pos=0.5, above]{$a$} (s_9)
		(s_8) edge node [pos=0.5,above]{$b$} (s_10);
		\end{tikzpicture}
		\caption{An illustrative example, the secret states are shaded}
		\label{fig:jo}

	\end{figure}
	
		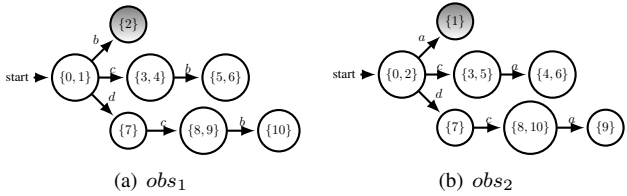
\begin{figure}
			\centering
			\subfigure[$obs_1$]{		
				\begin{tikzpicture}[shorten >=1pt,node distance=2cm,on grid,auto, bend angle=20, thick,scale=0.5, every node/.style={transform shape}] 
				\node[state,initial] (s0)   {$\{0,1\}$}; 
				\node[state,shade] (s1) [above right= of s0] {$\{2\}$}; 	
				\node[state] (s2) [right= of s0]  {$\{3,4\}$}; 
				\node[state] (s3) [below right= of s0] {$\{7\}$}; 
				\node[state] (s4)  [right= of s2] {$\{5,6\}$}; 
				\node[state] (s5) [right= of s3] {$\{8,9\}$}; 
				\node[state] (s6) [right= of s5]  {$\{10\}$}; 
				
				\path[->]

				(s0) edge node {$b$} (s1)  
				(s0) edge node {$c$} (s2)  
				(s0) edge node {$d$} (s3)  
				(s2) edge node {$b$} (s4)  
				(s3) edge node {$c$} (s5)  
				(s5) edge node {$b$} (s6) 
				
				; 
				\end{tikzpicture} 
			}
			\subfigure[$obs_2$]{		
				\begin{tikzpicture}[shorten >=1pt,node distance=2cm,on grid,auto, bend angle=20, thick,scale=0.5, every node/.style={transform shape}] 
				\node[state,initial] (s0)   {$\{0,2\}$}; 
				\node[state,shade] (s1) [above right= of s0] {$\{1\}$}; 	
				\node[state] (s2) [right= of s0]  {$\{3,5\}$}; 
				\node[state] (s3) [below right= of s0] {$\{7\}$}; 
				\node[state] (s4)  [right= of s2] {$\{4,6\}$}; 
				\node[state] (s5) [right= of s3] {$\{8,10\}$}; 
				\node[state] (s6) [right= of s5]  {$\{9\}$}; 
 	
				\path[->]

				(s0) edge node {$a$} (s1)  
				(s0) edge node {$c$} (s2)  
				(s0) edge node {$d$} (s3)  
				(s2) edge node {$a$} (s4)  
				(s3) edge node {$c$} (s5)  
				(s5) edge node {$a$} (s6)

				; 
				\end{tikzpicture} 
			}
			\caption{Observers for the system $G$.}
			\vspace{-3.5mm}
			\label{fig:obs}
		\end{figure}

\begin{example}		
Consider $G=(X,E,f,X_0)$ shown in Fig.~\ref{fig:jo}, where $E = \{a,b,c,d\}$. The secret states are  $X_{S}=\{1,2,6\}$, which are the shaded states in Fig.~\ref{fig:jo}. We assume that $G$ is observed by two intruders with different observation projections induced by $E_{o,1}=\{b,c,d\}$ and $E_{o,2}=\{a,c,d\}$, respectively. The observer $obs_1$ and $obs_2$ can be constructed in a standard way \cite{cassandras2008introduction} as shown in Fig.~\ref{fig:obs}. Each state in $obs_i$ contains the current state estimation of intruder $I_i$, $i=1,2$. From Fig.~\ref{fig:obs}, both observers reveal some secrets (the shaded states in Fig.~\ref{fig:obs}) without the opacity-enforcement mechanism. 
\end{example}

Since there is no coordination between the intruders, we follow the procedures in \cite{wu2016synthesis} to construct the all insertion structure (AIS) that encodes all the valid system and insertion function moves for each intruder respectively. It is then possible to extract an insertion function from the AIS.

The AIS $(Q,E_o\cup\{\epsilon\},f,q_0)$ can be seen as a game structure between the system and the insertion function, where $Q=Q_S\cup Q_I$, $Q_S$ denotes the system state set and $Q_I$ denotes the insertion function state set. Each $q\in Q_S$ has a pair of state estimates, the first one is the intruder's estimate, which could be wrong due to the inserted events, and the second estimate is the real system estimate. For each $q\in Q_I$, besides the intruder and system's state estimate, it also consists of current system output from $G$. $f(q,e)=q'$ for $q,q'\in Q$ and $e\in E_o\cup\{\epsilon\}$ represents the transition function, $q_0\in Q_S$ is the initial state. As shown in Fig. \ref{fig:ais}, the rectangles represent the system states  and the ellipses represent the insertion function states. All the transitions with events originated from the system states are system moves that are not controllable, while all the transitions with events from insertion function states are insertion function moves that the intruder actually observes.

\begin{theorem}\cite{wu2014synthesis}\label{thm:AIS}
CSO is privately enforceable if and only if the AIS is nonempty.
\end{theorem}

\begin{remark}
The main differences of our paper's AIS definition from \cite{wu2016synthesis} are two folds. The first is that we unfold the moves of the insertion function as well as the intruder's state estimate \emph{event by event}, while in \cite{wu2016synthesis}, the insertion function's move is from $E_o^*\cup\{\epsilon\}$, which could denote the whole string that has been inserted. For example, in our AIS definition, if we have a transition $q_0\xrightarrow{a}q_1\xrightarrow{b}q_2\xrightarrow{c}q_3$, where $q_0,q_1,q_2\in Q_I$ and $q_3\in Q_S$, in \cite{wu2016synthesis}'s definition, the same transition would be simplified to $q_0\xrightarrow{abc}q_3$. The second is that, if the system $G$ contains loops (for the simplest case, imagine there is a self-loop in some state), it could be the case that the inserted string contains $s^*$ for some $s\in E_o^*$ and becomes arbitrarily long (for example, Fig. 7 in \cite{wu2016synthesis}). In our paper, we restrict the inserted strings to be $*$-free, that is, the insertions cannot be arbitrarily long and we replace $s^*$ with $\epsilon$. Our definition with unfolding and $*$-free in insertions are to facilitate the analysis of joint opacity enforcement in Section \ref{sect:joint opacity}. 

\end{remark}
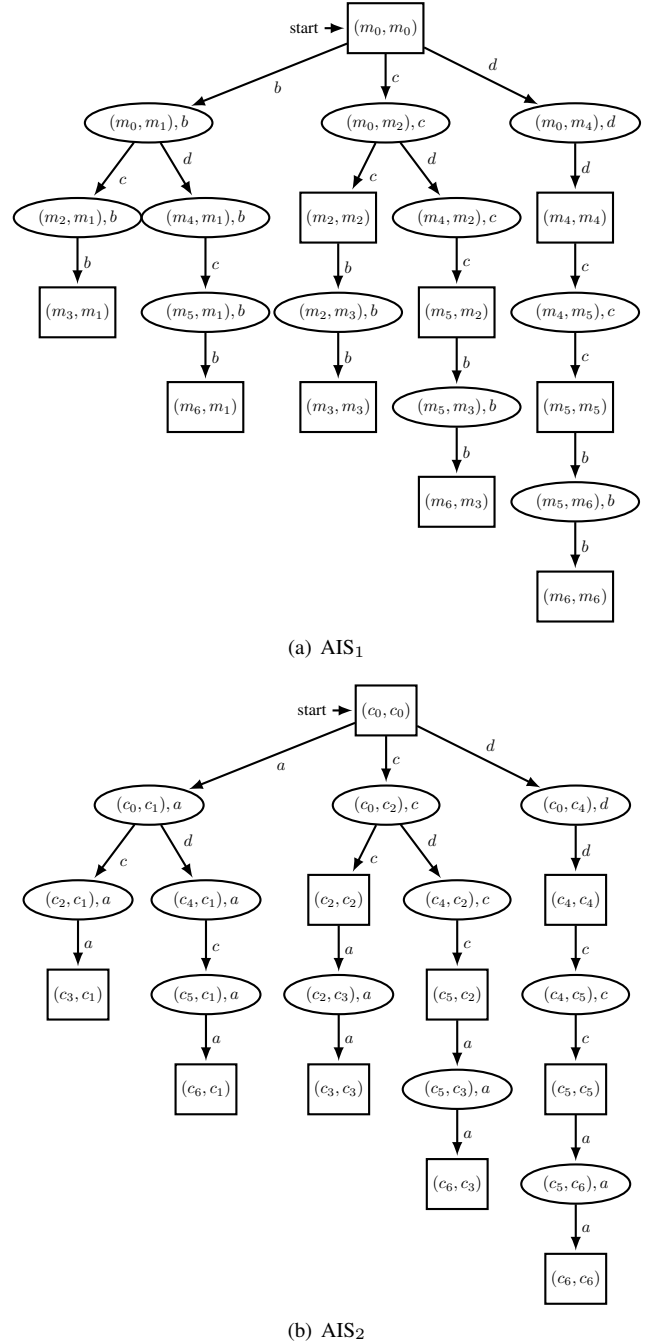
\begin{figure}[H]
	\centering
	\subfigure[AIS$_1$]{		
		\begin{tikzpicture}[shorten >=1pt,node distance=2cm,on grid,auto, bend angle=20, thick,scale=0.63, every node/.style={transform shape}] 
		\node[block,initial] (s0)   {$(m_0,m_0)$}; 
		\node[ellipse,draw] (s1) [below left = 2 cm and 5 cm of s0] {$(m_0,m_1),b$}; 	
		\node[ellipse,draw] (s2) [below = of s0]  {$(m_0,m_2),c$}; 
		\node[ellipse,draw] (s3) [below right= 2 cm and 4 cm of s0] {$(m_0,m_4),d$}; 
		
		\node[ellipse,draw] (s4) [below left= 2 cm and 1.5 cm of s1] {$(m_2,m_1),b$}; 
		\node[block] (s5) [below = of s4] {$(m_3,m_1)$}; 
		\node[ellipse,draw] (s6) [below right =2 cm and 1.2 cm of s1] {$(m_4,m_1),b$};
		\node[ellipse,draw] (s7) [below =  of s6] {$(m_5,m_1),b$};
		\node[block] (s8) [below = of s7] {$(m_6,m_1)$};  
		\node[block] (s9) [below left= 2 cm and 1 cm of s2] {$(m_2,m_2)$}; 
		\node[ellipse,draw] (s10) [below right = 2 cm and 1.5 cm of s2]  {$(m_4,m_2),c$}; 
		\node[block] (s11) [below = of s10]  {$(m_5,m_2)$}; 
		\node[block] (s12) [below of =  s3] {$(m_4,m_4)$}; 
		\node[ellipse,draw] (s13) [below =of s9] {$(m_2,m_3),b$};
		\node[block] (s14) [below =of s13] {$(m_3,m_3)$};
		\node[ellipse,draw] (s15) [below =of s11] {$(m_5,m_3),b$};
		\node[block] (s16) [below =of s15] {$(m_6,m_3)$};
		\node[ellipse,draw] (s17) [below =of s12] {$(m_4,m_5),c$};
		\node[block] (s18) [below =of s17] {$(m_5,m_5)$};
		\node[ellipse,draw] (s19) [below =of s18] {$(m_5,m_6),b$};
		\node[block] (s20) [below =of s19] {$(m_6,m_6)$};
		\path[->]
				(s0) edge node {$b$} (s1)  
				(s0) edge node {$c$} (s2)  
				(s0) edge node {$d$} (s3)  
				(s1) edge node {$c$} (s4)  
				(s1) edge node {$d$} (s6)  
				(s4) edge node {$b$} (s5) 
				(s6) edge node {$c$} (s7) 
				(s7) edge node {$b$} (s8) 
		
				(s2) edge node {$c$} (s9) 		
				(s9) edge node {$b$} (s13) 
				(s13) edge node {$b$} (s14) 
				(s2) edge node {$d$} (s10) 
				(s10) edge node {$c$} (s11) 				
				(s11) edge node {$b$} (s15) 
				(s15) edge node {$b$} (s16) 
				
				(s3) edge node {$d$} (s12) 
				(s12) edge node {$c$} (s17) 
				(s17) edge node {$c$} (s18) 				
				(s18) edge node {$b$} (s19) 
				(s19) edge node {$b$} (s20) 	
																					
		; 
		\end{tikzpicture} 
	}
	\subfigure[AIS$_2$]{		
		\begin{tikzpicture}[shorten >=1pt,node distance=2cm,on grid,auto, bend angle=20, thick,scale=0.63, every node/.style={transform shape}] 
		\node[block,initial] (s0)   {$(c_0,c_0)$}; 
		\node[ellipse,draw] (s1) [below left = 2 cm and 5 cm of s0] {$(c_0,c_1),a$}; 	
		\node[ellipse,draw] (s2) [below = of s0]  {$(c_0,c_2),c$}; 
		\node[ellipse,draw] (s3) [below right= 2 cm and 4 cm of s0] {$(c_0,c_4),d$}; 
		
		\node[ellipse,draw] (s4) [below left= 2 cm and 1.5 cm of s1] {$(c_2,c_1),a$}; 
		\node[block] (s5) [below = of s4] {$(c_3,c_1)$}; 
		\node[ellipse,draw] (s6) [below right =2 cm and 1.2 cm of s1] {$(c_4,c_1),a$};
		\node[ellipse,draw] (s7) [below =  of s6] {$(c_5,c_1),a$};
		\node[block] (s8) [below = of s7] {$(c_6,c_1)$};  
		\node[block] (s9) [below left= 2 cm and 1 cm of s2] {$(c_2,c_2)$}; 
		\node[ellipse,draw] (s10) [below right = 2 cm and 1.5 cm of s2]  {$(c_4,c_2),c$}; 
		\node[block] (s11) [below = of s10]  {$(c_5,c_2)$}; 
		\node[block] (s12) [below of =  s3] {$(c_4,c_4)$}; 
		\node[ellipse,draw] (s13) [below =of s9] {$(c_2,c_3),a$};
		\node[block] (s14) [below =of s13] {$(c_3,c_3)$};
		\node[ellipse,draw] (s15) [below =of s11] {$(c_5,c_3),a$};
		\node[block] (s16) [below =of s15] {$(c_6,c_3)$};
		\node[ellipse,draw] (s17) [below =of s12] {$(c_4,c_5),c$};
		\node[block] (s18) [below =of s17] {$(c_5,c_5)$};
		\node[ellipse,draw] (s19) [below =of s18] {$(c_5,c_6),a$};
		\node[block] (s20) [below =of s19] {$(c_6,c_6)$};
		\path[->]
		(s0) edge node {$a$} (s1)  
		(s0) edge node {$c$} (s2)  
		(s0) edge node {$d$} (s3)  
		(s1) edge node {$c$} (s4)  
		(s1) edge node {$d$} (s6)  
		(s4) edge node {$a$} (s5) 
		(s6) edge node {$c$} (s7) 
		(s7) edge node {$a$} (s8) 
		
		(s2) edge node {$c$} (s9) 		
		(s9) edge node {$a$} (s13) 
		(s13) edge node {$a$} (s14) 
		(s2) edge node {$d$} (s10) 
		(s10) edge node {$c$} (s11) 				
		(s11) edge node {$a$} (s15) 
		(s15) edge node {$a$} (s16) 
		
		(s3) edge node {$d$} (s12) 
		(s12) edge node {$c$} (s17) 
		(s17) edge node {$c$} (s18) 				
		(s18) edge node {$a$} (s19) 
		(s19) edge node {$a$} (s20) 	
		
		; 
		\end{tikzpicture} 
	}
	\caption{AISs for intruders $I_1$ and $I_2$.}
	\vspace{-3.5mm}
	\label{fig:ais}
\end{figure}
The AISs for intruders $I_1$ and $I_2$ in our motivating example are as shown in Fig \ref{fig:ais}, where $m_0=\{0,1\},m_1=\{2\},m_2=\{3,4\},m_3=\{5,6\},m_4=\{7\},m_5=\{8,9\},m_6=\{10\}$, $c_0=\{0,2\},c_1=\{1\},c_2=\{3,5\},c_3=\{4,6\},c_4=\{7\},c_5=\{8,10\},c_6=\{9\}$. For instance, in AIS$_1$ shown in Fig. \ref{fig:ais} (a), starting from the initial state where the intruder and the system's estimates are $(m_0,m_0)$. If the event $b$ occurs in the system, AIS$_1$ transits to the insertion function state $((m_0,m_1),a)$ since the system observer sees the event $b$ and the intruder observer observes nothing as the insertion has not been decided yet. Then if the insertion function decides to insert $c$, the system transits to the insertion function state $((m_2,m_1),c)$ as the intruder observer observes $c$ and the system observer will ignore the insertion function outputs. Then the real system output $b$ is appended and consequently the AIS transits to the system state $(m_3,m_1)$.

\begin{theorem}
	Given the system $G$ and $N$ intruders with observation mask $P_i$, $i\in\mathcal{N}$, D-CSO is privately enforceable if and only if AIS$_i$ is nonempty for all $i\in\mathcal{N}$.
\end{theorem}
\begin{proof}
On the one hand, D-CSO holds if and only if local CSO holds for all $ i\in \mathcal{N}$. On the other hand, for each intruder $I_i$, local CSO is privately enforceable if and only if the AIS$_i$ is not empty by Theorem \ref{thm:AIS}. Therefore the proof is completed.
\end{proof}

\section{Synthesis of Insertion Functions for Joint Opacity Enforcement}\label{sect:joint opacity}

Rather than observing the same system without coordination, in many applications, intruders do coordinate among themselves by exchanging their estimates of the system's states. For these applications, decentralized opacity notions of coordinated intruders need further investigation.

\begin{figure}[h]
\centering
\includegraphics[scale=0.4]{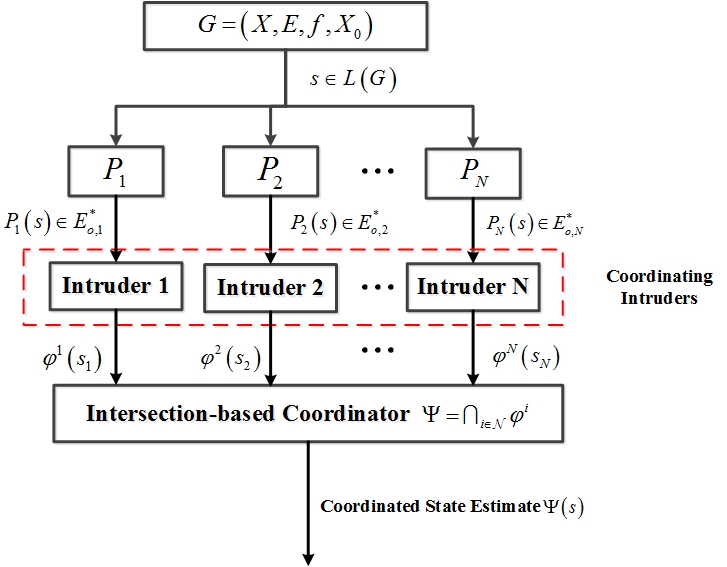}\\
\caption{The DES $G$ observed by intruders with intersection-based coordination protocols.}\label{dec_cor}
\vspace{-5mm}
\end{figure}

\subsection{Intersection-based Coordination Protocol}
In this section, we investigate intruders that may coordinate with each other via an {\it intersection-based} protocol \cite{paoli2012decentralized}. As shown in Fig.~\ref{dec_cor}, we assume that the team of intruders $I_i$, $i\in\mathcal{N}$ not only generate local state estimates but report the estimates to a {\it coordinator} as well. The coordinator has no knowledge about the system. It forms the so-called {\it coordinated estimate} by taking the intersection of the local estimates it receives. The communication from the local intruders to the coordinator is assumed to have no delay. The collaboration is restricted by the following rules: (1) intruders have no knowledge of the projections of one another; (2) the only collaboration between the intruders is through the coordinator.

Before proceeding to opacity issues in the coordinated decentralized framework, we first study the intersection-based coordination protocol in Fig.~\ref{dec_cor}. For the intruder $I_i$, $i\in\mathcal{N}$, a string-based local estimation map $\psi^i: P_i[L(G)]\to 2^X$ is defined as follows: for $s\in L(G)$ and $s_i:=P_i(s)$, $\psi^i(s_i)=f\left(x_0, P_i^{-1}(s_i)\cap L(G)\right)$.

Then, we define an {\it intersection-based} coordination protocol $\Psi: \Pi_{i\in\mathcal{N}} P_i[L(G)] \to 2^X$ as
\begin{equation}\label{eqn:int}
\Psi(s_1,s_2,\cdots,s_N)=\bigcap_{i\in\mathcal{N}}\psi^i(s_i)
\end{equation}

Intuitively, the coordination protocol $\Psi$ takes the intersection of the local estimates reported by the intruders and forms a {\it coordinated estimate} accordingly.

\subsection{Enforcement Scheme of Joint Opacity in Discrete-event Systems}
We now consider opacity issues of DES that can be observed by intruders following the intersection-based coordination protocol in Eq.~\ref{eqn:int}. Roughly speaking, the system is said to be {\it jointly current-state opaque} if no coordinated estimate ever reveals the secret information. 

\begin{definition}[Joint CSO (J-CSO)]
Given the set of observable events $E_{o,i}\subseteq E$ for intruder $I_i$, $i\in \mathcal{N}$, the secret state set $X_S$, the non-secret state set $X_{NS}$ and the intersection-based coordination protocol $\Psi$, the system $G=(X,E,f,X_0)$ is said to be jointly current-state opaque with respect to $E_{o,i}$, $i\in \mathcal{N}$, $X_S$ and $X_{NS}$ if for each intruder, local CSO holds and
\begin{equation}
\Psi(s_1,s_2,\cdots,s_N)\cap X_S\ne\emptyset \land \Psi(s_1,s_2,\cdots,s_N)\cap X_{NS}\ne\emptyset
\end{equation}
\end{definition}

In this paper, we present a centralized approach to synthesize the individual insertion functions to enforce J-CSO. The following example shows that, in general, local insertion functions that enforce D-CSO of a system may not enforce J-CSO.

\begin{example}\label{exp:D-CSO does not imply J-CSO}
With AIS$_1$ and AIS$_2$ in Fig. \ref{fig:ais}, D-CSO is guaranteed in Example 1.  However, if the two intruders can send their estimates to the intersection-based coordinator, joint opacity may be violated. For instance, from Fig.~\ref{fig:ais}, if the string $cab$ happens in the system, it will be projected to be $cb$ and $ca$ for intruders $1$ and $2$, respectively. If both insertion functions choose not to insert anything, which are valid moves from their local AISs, the resulting estimates reported by intruders $1$ and $2$, after observing $cb$ and $ca$, are $\{5,6\}$ and $\{4,6\}$, respectively. As our coordinator performs the intersection of the estimation, it will result in $\{6\}\in X_S$, which reveals a secret. 
\end{example}

Example \ref{exp:D-CSO does not imply J-CSO} implies that insertion functions that enforce D-CSO do not necessarily guarantee J-CSO. Therefore, the insertion functions need to be specifically coordinated to enforce the J-CSO. 

Our first step is to encode the AIS into a corresponding Nondeterministic Finite-state Mealy machine (NFM) for a concise representation.  
\begin{definition}
	An NFM is a 5-tuple 
	\begin{equation}
	 \mathcal{M}=(Q,\Sigma_{In},\Sigma_{Out},q_0,f_{NMF}),   
	\end{equation}
	where $Q$ is the set of states, $\Sigma_{In}$ and $\Sigma_{Out}$ are the sets of input and output symbols, respectively, $q_0\in Q$ is the initial state, $f_{NMF}(q,e)=(q',o)$ defines the transition and input output relation for $q,q'\in Q, e\in\Sigma_{In},o\in\Sigma_{Out}$. 
\end{definition}

The nondeterminism of an NFM comes from the fact that in general $|f_{NMF}(q,e)|\ge 1$, which implies that the same input on the same state may result in non-unique insertions and transit to different states. Our NFM formulation is similar to the insertion automaton \cite{wu2016synthesis} but we allow nondeterministic choices of insertions upon observing a system output $e\in E_o$. The procedure to convert an AIS into an NFM $\mathcal{M}=(Q,\Sigma_{In},\Sigma_{Out},q_0,f)$ is as follows. $Q$ is the set of all the systems states of AIS. $\Sigma_{In}$ is the set of all the events from system states and $\Sigma_{Out}$ is the set of all the possible insertion strings. The transition function is defined as $f(q,e)=(q',o)$, where $o = o'+e, o'\in E_I^*, e\in E_o$, $o'$ is the inserted string, $e$ is the system input and $o$ denotes the output from the state $q$ when the system input is $e$. 

\begin{example}
Fig.~\ref{fig:nfm} denotes the NFMs corresponding to AIS$_1$ and AIS$_2$ in Fig.~\ref{fig:ais}, respectively. Note that, different from AIS, in the NFM formulation, upon observing an event $e$, the state directly jumps from $q$ to $q'$ while outputting the string $o$. However, what really happens, as shown in AIS, is that the intruder's estimation is updated event by event for each output of the insertion function. Such estimation evolution is omitted in the NFM formulation for conciseness but can be recovered from our AIS. 
\end{example}
\begin{figure}[ht]
	\centering
	\subfigure[$NFM_1$]{		
		\begin{tikzpicture}[shorten >=1pt,node distance=2cm,on grid,auto, bend angle=20, thick,scale=0.7, every node/.style={transform shape}] 
		\node[block,initial above] (s0)   {$(m_0,m_0)$}; 
		\node[block] (s1) [left = 3cm  of s0] {$(m_3,m_1)$}; 	
		\node[block] (s2) [below left = 2 cm and 3 cm of s0]  {$(m_6,m_1)$}; 
		\node[block] (s3) [right = 3 cm of s2] {$(m_2,m_2)$}; 
		
		\node[block] (s4) [right = 3 cm of s3] {$(m_5,m_2)$}; 
		\node[block] (s5) [right = 6 cm of s0] {$(m_4,m_4)$}; 
		\node[block] (s6) [below =  of s3] {$(m_3,m_3)$};
		\node[block] (s7) [below =  of s4] {$(m_6,m_3)$};
		\node[block] (s8) [below =  of s5] {$(m_5,m_5)$};  
		\node[block] (s9) [below =  of s8] {$(m_6,m_6)$}; 
		
		\path[->]
				(s0) edge node [above] {$b/cb$} (s1)  
				(s0) edge node [above] {$b/dcb$} (s2)  
				(s0) edge node {$c/c$} (s3)  
				(s0) edge node {$c/dc$} (s4)  
				(s0) edge node {$d/d$} (s5)  
				(s3) edge node {$b/b$} (s6) 
				(s4) edge node {$b/b$} (s7) 
				(s5) edge node {$c/c$} (s8) 
				(s8) edge node {$b/b$} (s9)

		; 
		\end{tikzpicture} 
	}
	\subfigure[$NFM_2$]{		
		\begin{tikzpicture}[shorten >=1pt,node distance=2cm,on grid,auto, bend angle=20, thick,scale=0.65, every node/.style={transform shape}] 
		\node[block,initial above] (s0)   {$(c_0,c_0)$}; 
		\node[block] (s1) [left = 3cm  of s0] {$(c_3,c_1)$}; 	
		\node[block] (s2) [below left = 2 cm and 3 cm of s0]  {$(c_6,c_1)$}; 
		\node[block] (s3) [right = 3 cm of s2] {$(c_2,c_2)$}; 
		
		\node[block] (s4) [right = 3 cm of s3] {$(c_5,c_2)$}; 
		\node[block] (s5) [right = 6 cm of s0] {$(c_4,c_4)$}; 
		\node[block] (s6) [below =  of s3] {$(c_3,c_3)$};
		\node[block] (s7) [below =  of s4] {$(c_6,c_3)$};
		\node[block] (s8) [below =  of s5] {$(c_5,c_5)$};  
		\node[block] (s9) [below =  of s8] {$(c_6,c_6)$}; 
		
		\path[->]
				(s0) edge node [above] {$a/ca$} (s1)  
				(s0) edge node [above] {$a/dca$} (s2)  
				(s0) edge node {$c/c$} (s3)  
				(s0) edge node {$c/dc$} (s4)  
				(s0) edge node {$d/d$} (s5)  
				(s3) edge node {$a/a$} (s6) 
				(s4) edge node {$a/a$} (s7) 
				(s5) edge node {$c/c$} (s8) 
				(s8) edge node {$a/a$} (s9)

		; 
		\end{tikzpicture} 
	}
	\caption{The NFMs converted from AISs.}
	\vspace{-3.5mm}
	\label{fig:nfm}
\end{figure}
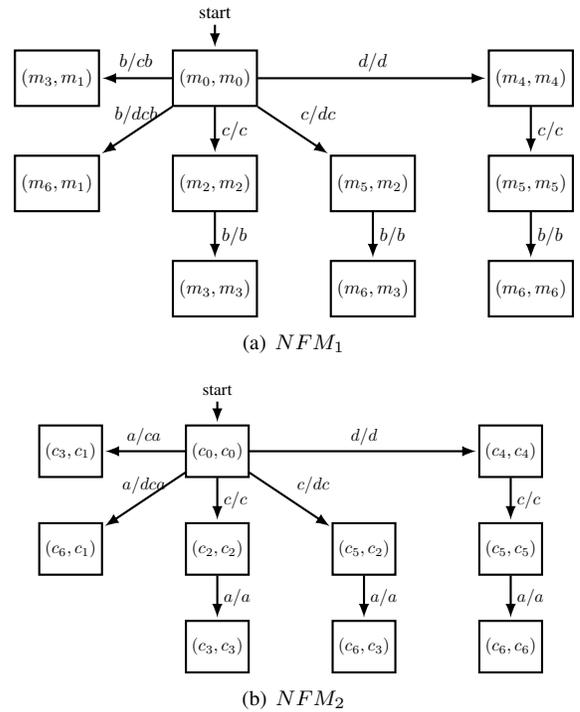

To keep the NFMs synchronized with the original system that intruders try to compromise, we construct another system observer $obs$ as a DFA with $E_{obs,o} = \bigcup_{i\in\mathcal{N}} E_{o,i}$. That is, if an event is observable to any one of the intruders, it is observable to this system observer. In our example, $E_{obs,o}=E$ and the observer has the identical structure with the original system as shown in Fig.~\ref{fig:jo}. The observer $obs$ can be viewed as an NFM that outputs empty string $\epsilon$ for all inputs. Given $N$ AISs' in the form of NFM $\mathcal{M}_i=(Q_i,\Sigma_{In}^i,\Sigma_{Out}^i,q^i_0,f_i)$ for $i\in\mathcal{N}$ and the system observer $obs$, we can obtain the composed NFM $\mathcal{G}=(Q_G,\Sigma_{In},\Sigma_{Out},q_0,f)$ that describes all the possible combined insertion behaviors,  where $Q= Q_{1}\times Q_2 ... \times Q_{N}\times Q_{obs}$, $\Sigma_{In}=\Sigma_{obs,In}=E_{obs,o}$, $\Sigma_{Out}=\Sigma_{1,Out}\times\Sigma_{2,Out}...\times\Sigma_{N,Out}\times\{\epsilon\}$. The transition relation $f((q_1,q_2,...,q_N,q_{obs}),e)=((q'_1,q'_2,...,q'_N,q'_{obs}),o)$ is given by 
\begin{itemize}
	\item $o=(o_1,...,o_N,\epsilon)$
	\item $(q'_{obs},\epsilon)=f_{obs}(q_{obs},e)$
	\item $(q'_{i},o_i)=f_i(q_{i},e)$ if $e\in E_{o,i}$; otherwise $(q'_{i},o_i)=(q_i,\epsilon)$
\end{itemize}

While constructing this product NFM $\mathcal{G}$, we assume that when an event $e\in\Sigma_o$ in the system occurs, it is guaranteed that for every intruder $I_i$ such that $e\in E_{o,i}$ holds, its corresponding insertion function will finish outputting the modified string $o_i$ before the next system event $e'\in E_o$ is generated. It is always possible since we restrict the output of the insertion functions to be $*$-free. In this regard, every insertion function is synchronized with the system inputs. 

\begin{example}
Fig.~\ref{fig:nfm_all} illustrates the $\mathcal{G}$ from Example 2 and the corresponding $NFM_1$ and $NFM_2$ in Fig.~\ref{fig:nfm}. For simplicity, in this figure we omit the constant output $\epsilon$ from the system observer as well as each individual observer's state estimation.
\end{example}

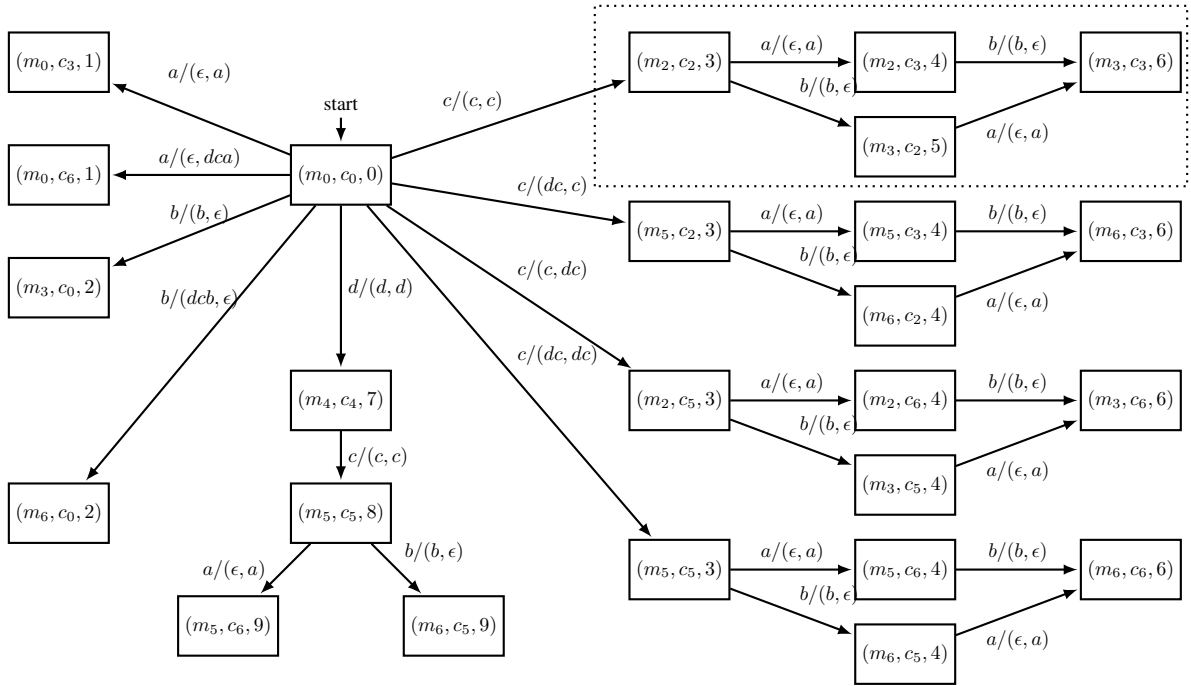
\begin{figure*}
    \centering
		\begin{tikzpicture}[shorten >=1pt,node distance=2cm,on grid,auto, bend angle=20, thick,scale=0.75, every node/.style={transform shape}] 
		\node[block,initial above] (s0)   {$(m_0,c_0,0)$}; 
		\node[block] (s1) [above left  = 2cm and 5 cm  of s0] {$(m_0,c_3,1)$}; 	
		\node[block] (s2) [below = 2cm of s1]  {$(m_0,c_6,1)$}; 
		\node[block] (s3) [below = 2cm of s2] {$(m_3,c_0,2)$}; 
		
		\node[block] (s4) [below = 4cm of s3] {$(m_6,c_0,2)$}; 
		\node[block] (s5) [below = 4cm of s0] {$(m_4,c_4,7)$}; 
		\node[block] (s6) [below =  of s5] {$(m_5,c_5,8)$};
		\node[block] (s7) [below left = 2 cm and 2cm of s6] {$(m_5,c_6,9)$};
		\node[block] (s8) [below right = 2cm and 2cm of s6] {$(m_6,c_5,9)$};  
		\node[block] (s9) [above right = 2 cm and 6 cm of s0] {$(m_2,c_2,3)$}; 
		\node[block] (s10) [below = 3 cm of s9] {$(m_5,c_2,3)$}; 
		\node[block] (s11) [below = 3 cm  of s10] {$(m_2,c_5,3)$};
		\node[block] (s12) [below = 3 cm  of s11] {$(m_5,c_5,3)$};
		\node[block] (s13) [right = 4 cm of s9] {$(m_2,c_3,4)$};  
		\node[block] (s14) [below = 1.5 cm of s13] {$(m_3,c_2,5)$}; 
		\node[block] (s15) [below = 1.5 cm of s14] {$(m_5,c_3,4)$}; 
		\node[block] (s16) [below = 1.5 cm of s15] {$(m_6,c_2,4)$}; 
		\node[block] (s17) [below = 1.5 cm of s16] {$(m_2,c_6,4)$}; 
		\node[block] (s18) [below = 1.5 cm of s17] {$(m_3,c_5,4)$}; 
		\node[block] (s19) [below = 1.5 cm of s18] {$(m_5,c_6,4)$}; 
		\node[block] (s20) [below = 1.5 cm of s19] {$(m_6,c_5,4)$}; 
		\node[block] (s21) [right = 4 cm of s13] {$(m_3,c_3,6)$}; 
		\node[block] (s22) [below = 3 cm of s21] {$(m_6,c_3,6)$}; 
		\node[block] (s23) [below = 3 cm of s22] {$(m_3,c_6,6)$}; 
		\node[block] (s24) [below = 3 cm of s23] {$(m_6,c_6,6)$}; 
		\path[->]
				(s0) edge node [above=0.5cm] {$a/(\epsilon,a)$} (s1)  
				(s0) edge node [above] {$a/(\epsilon,dca)$} (s2)  
				(s0) edge node [above] {$b/(b,\epsilon)$} (s3)  
				(s0) edge node [above=0.5cm] {$b/(dcb,\epsilon)$} (s4)  
				(s0) edge node {$d/(d,d)$} (s5)  
				(s5) edge node {$c/(c,c)$} (s6) 
				(s6) edge node [left = 0.2 cm] {$a/(\epsilon,a)$} (s7) 
				(s6) edge node {$b/(b,\epsilon)$} (s8) 
				(s0) edge node {$c/(c,c)$} (s9) 
				(s0) edge node {$c/(dc,c)$} (s10) 
				(s0) edge node {$c/(c,dc)$} (s11) 
				(s0) edge node {$c/(dc,dc)$} (s12) 
				(s9) edge node {$a/(\epsilon,a)$} (s13) 
				(s9) edge node {$b/(b,\epsilon)$} (s14) 
				(s10) edge node {$a/(\epsilon,a)$} (s15) 
				(s10) edge node {$b/(b,\epsilon)$} (s16) 
				(s11) edge node {$a/(\epsilon,a)$} (s17) 
				(s11) edge node {$b/(b,\epsilon)$} (s18) 
				(s12) edge node {$a/(\epsilon,a)$} (s19) 
				(s12) edge node {$b/(b,\epsilon)$} (s20) 
				(s13) edge node {$b/(b,\epsilon)$} (s21) 
				(s14) edge node [below = 0.2cm]{$a/(\epsilon,a)$} (s21) 
				(s15) edge node {$b/(b,\epsilon)$} (s22) 
				(s16) edge node [below = 0.2cm]{$a/(\epsilon,a)$} (s22) 
				(s17) edge node {$b/(b,\epsilon)$} (s23) 
				(s18) edge node [below = 0.2cm]{$a/(\epsilon,a)$} (s23) 
				(s19) edge node {$b/(b,\epsilon)$} (s24) 
				(s20) edge node [below = 0.2cm]{$a/(\epsilon,a)$} (s24) 
	
		; 
\draw[thick,dotted]     (4.5,-0.2) rectangle (15,3);
		\end{tikzpicture} 
    \caption{The NFM $\mathcal{G}$, the states in the dashed box are pruned}
    \vspace{-3.5mm}
    \label{fig:nfm_all}
\end{figure*}

For a given transition $f((q_0,q_1,...,q_{n-1},q_{obs}),e)=((q'_0,q'_1,...,q'_{n-1},q'_{obs}),o)$, it is then possible to check whether the secrets will be revealed and joint opacity could be violated during this transition with the help of the AISs, since as mentioned earlier, the event by event evolution of the state estimation for each intruder upon observing a modified string is omitted in the NFM but not AIS. 

For example, in the NFM from Fig. \ref{fig:nfm_all}, starting from the initial state, when the event $c$ happens and the insertion functions decide to insert $d$ and $\epsilon$ respectively, the transition is $(m_0,c_0,0)\xrightarrow{c/(dc,c)}(m_5,c_2,3)$. From the AISs in Fig. \ref{fig:ais}, the evolution of each intruder's estimation can be see as a two step transition
$(m_0,c_0,0)\xrightarrow{d,c}(m_4,c_2,3)\xrightarrow{c,\epsilon}(m_5,c_2,3)$. Note that there is an intermediate state $(m_4,c_2,3)$ that is not shown in $\mathcal{G}$. In the first step, upon observing the system event $c$, the first insertion function outputs $d$ and the second insertion function, since it decides to insert nothing, the system event $c$ is directly outputted. Therefore, the estimations evolve from $m_0$ to $m_4$, $c_0$ to $c_2$, and the system observer's estimation changes from $0$ to $3$. In the second step, the first insertion function outputs the system event $c$ and the second insertion function outputs $\epsilon$. It can be seen that our assumption is that the event output (including $\epsilon$) for each intruder is synchronized.

To determine if a transition $(q,e)\rightarrow (q',o)$ in $\mathcal{G}$ is safe, we examine every intermediate state from the AISs that evolves with each output event to see if the joint estimation reveals a secret. By definition, $q_{obs}$ encodes the \emph{true} set of states that the system is currently in. The first element of  $AIS_i$ state --- we denote as $est_i$ --- represents each \emph{intruder's estimation} of the current state. According to our coordination rule, the joint estimation $J_{est}$ is then obtained by taking intersection among $q_{obs}$ and $est_i,i=1,\ldots,N$. 
$$
J_{est} = q_{obs}\cap est_0\cap ... \cap est_{n-1}
$$

\begin{proposition}
	$J_{est}$ does not reveal a secret if  $J_{est}\cap X_S\neq\emptyset\implies J_{est}\cap X_{NS}\neq\emptyset$
\end{proposition}

We now define J-CSO in the presence of synthesized insertion functions as follows.
\begin{definition}
	Given $N$ intruders with unobservable event sets $E_{uo,1},E_{uo,2},...,E_{uo,N}$ and their insertion functions $f_{I,1},f_{I,2},...,f_{I,N}$, the system is J-CSO against the intruders if
	\begin{itemize}
		\item For each individual intruder $i$, the insertion function $f_{I,i}$ enforces local CSO.
		\item The $J_{est}$ never reveals the secret.
	\end{itemize}
	Furthermore, we define J-CSO to be jointly privately enforceable if all individual insertion functions are locally privately enforcing and $J_{est}$ never reveals the secret.
\end{definition}

An intermediate state is unsafe if its $J_{est}$ reveals a secret. A transition $(q,e)\rightarrow (q',o)$ in $\mathcal{G}$ is unsafe if any of the intermediate states between $q$ and $q'$ is unsafe. Similarly, any state $q\in Q$ of $\mathcal{G}$ is unsafe if its $J_{est}$ reveals a secret. If a transition is found to be unsafe, it will be pruned. If a state $q\in Q$ is found to be unsafe, this state, together with all its incoming and outgoing transitions, will be pruned. If after the pruning, at some state $q'\in Q$, there is no incoming transition (except the initial state) or there is no outgoing transition defined on an event $e$ that could happen in this state, which implies that the system blocks when $e$ happens at $q'$ since there is no insertion function available, then such state is also unsafe and all its incoming and outgoing transitions will be pruned. Again, such pruning may trigger new deadlocks and create unsafe states. Therefore, this is an iterative process until no unsafe state is found or the initial state is pruned. 

For example, as shown in Fig. \ref{fig:nfm_all}, the state $(m_3,c_3,6)$ is an unsafe state that reveals the secret. Because $m_3\cap c_3\cap \{6\}=\{5,6\}\cap\{4,6\}\cap\{6\}=\{6\}\in X_S$. Therefore it has to be pruned, which results in the states $(m_2,c_3,4)$ and $(m_3,c_2,5)$ being unsafe since there are no outgoing transitions any more. Consequently, pruning $(m_2,c_3,4)$ and $(m_3,c_2,5)$ makes $(m_2,c_2,3)$ unsafe. After deleting $(m_2,c_2,3)$, the pruning process stops. Since no more state or transitions is found to be unsafe, the resulting $\mathcal{G}$ can be found in Fig. \ref{fig:nfm_all}, excluding the states in the dashed box.

\begin{theorem}
	Given the system model $G$ and $N$ intruders with observation projections $P_i$, $i\in\mathcal{N}$, J-CSO is jointly privately enforceable if and only if $\mathcal{G}$ is nonempty after pruning.
\end{theorem}
\begin{proof}
If J-CSO is jointly privately enforceable, in our definition, it implies that for each individual intruder $I_i$, local opacity is privately enforceable and thus AIS$_i$ is nonempty by Theorem \ref{thm:AIS}. Since AIS$_i$ encodes all the possible local insertion functions that are privately enforcing, $\mathcal{G}$ as the product of AISs encodes all the possible joint insertion strategies that are privately enforcing. Since J-CSO is jointly privately enforceable, there exists at least one local insertion function for each intruder $I_i$ that is privately enforcing and the joint estimate never reveals the secret. Thus the joint insertion strategy is nonempty, which implies that $\mathcal{G}$ is nonempty.

Conversely, non-emptiness of $\mathcal{G}$ implies that AIS$_i$ is nonempty for any $i$. Thus, the local opacity is guaranteed. Furthermore, since $\mathcal{G}$, after pruning, encodes all the valid privately enforcing insertion functions for each intruder such that the joint state estimate never reveals the secret, J-CSO is guaranteed.
\end{proof}
\subsection{Complexity Analysis}
Given $N$ intruders and the system $G$ with $|X|$ states, the space and time complexity to construct each AIS is polynomial with $|X_\mathcal{E}|$ \cite{wu2016synthesis}, where $|X_\mathcal{E}|=2^{|X|}$ denotes the total number of states of the state estimator. Each NFM's state space is at most the state pace of its AIS. Therefore, the space complexity to construct $\mathcal{G}$ is polynomial in $|X_\mathcal{E}|$ and exponential in $|\mathcal{N}|$. The pruning process, in the worst case, looks over all the states in $\mathcal{G}$ and intermediate states, which is also polynomial in $|X_\mathcal{E}|$ and exponential in $|\mathcal{N}|$. So to sum up, the space and time complexity in our proposed centralized synthesis approach are both polynomial with $|X_\mathcal{E}|$ and exponential in $|\mathcal{N}|$.
\section{Conclusion}
In this paper, we investigate the opacity-enforcement problem for discrete-event systems that can be observed by multiple intruders. The major contribution of this paper is summarized as follows. First, we introduce opacity notions for two cases of DES in the presence of multiple intruders, one with coordination and the other without coordination. Next, we adopt the event insertion mechanism to ensure decentralized opacity for intruders without coordination; the synthesized insertion functions are further refined to enforce joint opacity of DES when intruders can coordinate via an intersection-based protocol. Future research directions may include: (i) introducing notions of joint opacity corresponding to other types of coordination protocols among the intruders; (ii) development of algorithms for enforcing other notions of joint opacity with respect to the new types of coordination protocols.

\bibliographystyle{IEEEtran}        
\bibliography{OpacityEnforcement}

\end{document}